\documentclass[a4paper,UKenglish]{lipics}
\usepackage{microtype}%if unwanted, comment out or use option "draft"

\title{Towards efficient decoding of classical-quantum polar codes}
\titlerunning{Towards efficient decoding of classical-quantum polar codes} %optional, in case that the title is too long; the running title should fit into the top page column

\author[1]{Mark M. Wilde}
\author[2]{Olivier Landon-Cardinal}
\author[1]{Patrick Hayden}
\affil[1]{School of Computer Science, McGill University\\
  3480 University Street, Montreal, Quebec H3A 2A7, Canada\\
  \texttt{mwilde@gmail.com; patrick@cs.mcgill.ca}}
\affil[2]{D\'epartement de Physique, Universit\'e de Sherbrooke\\
   Sherbrooke, Qu\'ebec J1K 2R1, Canada\\
   \texttt{olivier.landon-cardinal@usherbrooke.ca}}
\authorrunning{M.~M.~Wilde, O.~Landon-Cardinal, and P.~Hayden} %mandatory. First: Use abbreviated first/middle names. Second (only in severe cases): Use first author plus 'et. al.'

\Copyright[by-nc-nd]{Mark M. Wilde, Olivier Landon-Cardinal, and Patrick Hayden}%mandatory. Default is "by";  http://creativecommons.org/licenses/by/3.0/

% mandatory: Please choose ACM 1998 classifications from http://www.acm.org/about/class/ccs98-html . E.g., cite as "F.1.1 Models of Computation". 
% mandatory: Please provide 1-5 keywords
\serieslogo{}%please provide filename (without suffix)
\volumeinfo%(easychair interface)
  {}% editors
  {2}% number of editors: 1, 2, ....
  {$8^{\text{th}}$ Conference on the Theory of Quantum Computation, Communication, and Cryptography}% event
  {1}% volume
  {1}% issue
  {1}% starting page number
\EventShortName{}
\DOI{10.4230/LIPIcs.xxx.yyy.p}% to be completed by the volume editor
\let\originalleft\left
\let\originalright\right
\def\left#1{\mathopen{}\originalleft#1}
\def\right#1{\originalright#1\mathclose{}}

\begin{document}

\maketitle

\begin{abstract}
%Prior results of Holevo, Schumacher, Westmoreland, and others have suggested
%that collective measurements are required in order for a sender and receiver
%to communicate at the Holevo rate and for the receiver to reliably decode all
%of the information bits that the sender transmits over a channel with
%classical inputs and quantum outputs (cq channel). 
Known strategies for sending bits at the capacity rate over a general channel with classical input and quantum output (a cq channel) require the decoder to implement impractically complicated collective measurements.
Here, we show that
a fully collective strategy is not necessary in order to recover all of the
information bits. In fact, when coding for a large number $N$ uses of a cq channel $W$, 
$N\cdot I\left(  W_{\text{acc}}\right)  $ of the bits can
be recovered by a non-collective strategy which amounts to coherent quantum processing
of the results of product measurements, where
$I\left(  W_{\text{acc}}\right)  $ is the accessible information of the
channel~$W$. In order to decode the
other $N\left(  I\left(  W\right)  -I\left(  W_{\text{acc}}\right)  \right)  $
bits, where $I\left(  W\right)  $ is the Holevo rate, our conclusion is that the receiver should
employ collective measurements. We also present two other
results:\ 1)\ collective Fuchs-Caves measurements (quantum likelihood ratio
measurements) can be used at the receiver to achieve the Holevo rate and
2)\ we give an explicit form of the Helstrom measurements used in small-size
polar codes. The main approach used to demonstrate these results is a quantum
extension of Arikan's polar codes.

\end{abstract}

\section{Introduction}

One of the most impressive recent developments in coding theory is the theory
of polar codes~\cite{A09}. These codes are provably capacity achieving, and
their encoding and decoding complexities are both $O\left(  N\log N\right)  $,
where $N$ is the number of channel uses. Polar codes are based on the channel
polarization effect, in which a recursive encoding induces a set of $N$
synthesized channels from $N$ instances of the original channel, such that
some of the synthesized channels are nearly perfect and the others are nearly
useless. The fraction of synthesized channels that is nearly perfect is equal to
the capacity of the channel, and thus the coding scheme is simple:\ send the
information bits through the synthesized channels that are nearly perfect.

An essential component of the polar coding scheme is Arikan's successive
cancellation decoding algorithm \cite{A09}. This algorithm is channel
dependent and operates as its name suggests:\ it decodes the information bits
one after another, using previously decoded information to aid in constructing
a test for decoding each bit in succession. In particular, the test for
decoding each information bit is a likelihood ratio test. Due to the structure
in the polar encoder, there is a great deal of  structure in the decoding tests, so much
so that each likelihood ratio can be recursively computed.\ The upshot is that
the complexity of the decoding algorithm is $O\left(  N\log N\right)  $.

Recently, there has been some effort in extending the theory of polar coding
to the problem of transmission over quantum
channels \cite{WG11,RDR11,WR12,WR12a}. In particular, these works developed
the theory of polar coding for transmitting classical data over an arbitrary
quantum channel \cite{WG11}, private classical data over an arbitrary quantum
channel \cite{WR12a}, quantum data over a quantum Pauli or erasure channel
\cite{RDR11}, and quantum data over an arbitrary quantum channel \cite{WR12}.
To prove that the polar coding schemes in Refs.~\cite{WG11,WR12,WR12a} achieve
communication rates equal to well-known formulas from quantum information
theory, the authors of these works constructed a quantum successive
cancellation decoder as a sequence of quantum hypothesis tests (in the spirit
of Arikan \cite{A09}) and employed Sen's non-commutative union bound
\cite{S11}\ in the error analysis. The major question left open from this
effort is whether there exists an efficient implementation for a quantum
successive cancellation decoder.\footnote{By efficient, we mean that the decoder
should run in $O(N^2)$ time on a quantum computer (or even better $O(N\log N)$).
In computational complexity theory, ``efficient'' is often regarded to mean that
an algorithm runs in time polynomial in the input length. However, for the
demanding application of channel coding where delay should be minimized, we 
will consider a decoding algorithm to be ``efficient'' if it has a
near-linear running time.}$^{,}$\footnote{Note that the scheme from
Ref.~\cite{RDR11} \textit{does} provide an efficient $O(N \log N)$ implementation of a
quantum successive cancellation decoder, essentially because sending
classical states (encoded in some orthonormal basis)\ through a Pauli or
erasure channel induces an effectively classical channel at the output (such
that the resulting output states are commuting). One can then exploit a
coherent version of Arikan's successive cancellation decoder to decode quantum
information. Although this advance is useful, we would like to have an
efficient decoder for an \textit{arbitrary} quantum channel.}

In this paper, we detail our progress towards finding an
efficient quantum successive cancellation decoder. The decoder outlined here
is useful for decoding classical information transmitted over a channel with
classical inputs and quantum outputs (known as a \textquotedblleft
classical-quantum channel\textquotedblright\ or \textquotedblleft cq
channel\textquotedblright\ for short). Since the schemes for private classical
communication \cite{WR12a}\ and quantum communication \cite{WR12}\ rely on the
quantum successive cancellation decoder from Ref.~\cite{WG11}, our results
here have implications for these polar coding schemes as well. Our main result
can be stated succintly as follows:

\begin{claim}
\label{clm:main-result}In order to achieve the symmetric Holevo capacity
$I\left(  W\right)  $ of an arbitrary cq channel $W$, at most $N \left(  I\left(
W\right)  -I\left(  W_{\text{acc}}\right)  \right)  $ of the bits require a fully
collective strategy in order for them to be decoded reliably, while the other
$N\cdot I\left(  W_{\text{acc}}\right)  $ bits can be decoded efficiently and
reliably in time $O(  N^{2})  $ on a quantum computer using a
product strategy that amounts to coherent quantum processing of the outcomes
of product measurements.
\end{claim}

Although the main result of this paper might be considered modest in light of
reaching the full goal stated above, it still represents non-trivial progress
beyond prior research and towards answering the efficient polar decoding
question. Indeed, one might think that collective measurements would be necessary
in order to recover any of the bits of a message when communicating at the
Holevo capacity rate, as suggested by the original work of Holevo
\cite{Hol98}, Schumacher, and Westmoreland \cite{SW97} and follow-up efforts
on the pure-loss bosonic channel \cite{GGLMSY04,G11}. Even the recent sequential
decoding schemes suggest the same \cite{GLM10,S11} (see also \cite{WGTS11} for
the pure-loss bosonic case). As a side note, these sequential decoding schemes require a number of
measurements exponential in the number of channel uses---thus, even though the
physical realization of a single one of these measurements may be within
experimental reach \cite{OPJ12}, the fact that these schemes require an
exponential number of measurements excludes them from ever being practical.
The previous result in Ref.~\cite{WG11}\ suggests that only a linear number of
collective measurements are required to achieve the Holevo rate, and our work
here demonstrates that the number of collective measurements required is
at most $N\left(  I\left(
W\right)-I({W_\text{acc}})\right)  $.

This paper contains other results of interest. First, we prove that
collective Fuchs-Caves measurements (or quantum likelihood ratio measurements)
\cite{FC95} suffice for achieving the Holevo information rate with a cq polar
coding scheme. It was already known from Ref.~\cite{WG11} that a sequence of
Helstrom measurements suffices for achieving this rate, so this new result
just adds to the ways in which one can achieve the Holevo rate of
communication. We also plot the fraction of requisite collective measurements
as a function of the mean photon number of the signaling states for the case
of the pure-loss bosonic channel, in order to have a sense of the physical
requirements necessary for high-rate communication over this channel. As one
would expect, the fraction of collective measurements needed increases as the
mean photon number of the signaling states decreases---we expect this to
happen since the low photon-number regime is more quantum due to the non-orthogonality of the signaling states. Finally, we detail the
explicit form of a polar decoder that uses Helstrom measurements---we do this
for some simple two-, four-, and eight-bit polar codes. This final result should give a
sense of how one can specify these tests for larger blocklength polar codes.

The paper is organized as follows. The next section reviews background
material such as cq channels, the Holevo quantity, quantum fidelity, the
accessible information, and the classical fidelity (Bhattacharya parameter).
Section~\ref{sec:Fuchs-Caves} reviews the Fuchs-Caves measurement from
Ref.~\cite{FC95} and provides a useful upper bound on the error probability of
a hypothesis test that employs this measurement as the decision rule. We
review classical-quantum polar codes in Section~\ref{sec:polar-code-review}.
Our first simple observation is that collective Fuchs-Caves measurements
suffice for achieving the Holevo rate of communication
(Section~\ref{sec:baby-result}). Our main result, a justification for
Claim~\ref{clm:main-result}, appears in Section~\ref{sec:main-result}. In
Section~\ref{sec:bosonic}, we discuss the implications of
Claim~\ref{clm:main-result} for the pure-loss bosonic channel. Our last result
on the explicit form of the Helstrom decoder for two-, four-, and eight-bit polar codes
appears in Section~\ref{sec:last-result}. Finally, we conclude with a summary
of our results and suggest that the Schur transform might be helpful in
obtaining a general solution to the problem discussed in this paper.

\section{Preliminaries}

A classical-quantum channel (cq channel) has a classical input and a quantum output. In
this work, we only consider cq channels with binary inputs, written as%
\begin{equation}
W:x\rightarrow\rho_{x},
\end{equation}
where $W$ labels the channel, the input $x\in\left\{  0,1\right\}  $, and
$\rho_{x}$ is a density operator.
% satisfying $\rho_{x}\geq0$ and Tr$\left\{\rho_{x}\right\}  =1$. 
The symmetric Holevo information of this channel is%
\begin{equation}
I\left(  W\right)  \equiv H\left(  \left(  \rho_{0}+\rho_{1}\right)
/2\right)  -\left[  H\left(  \rho_{0}\right)  +H\left(  \rho_{1}\right)
\right]  /2,
\end{equation}
where $H\left(  \sigma\right)  \equiv-$Tr$\left\{  \sigma\log_{2}%
\sigma\right\}  $ is the von Neumann entropy. The symmetric Holevo information
gives one way to characterize the quality of a cq channel for data
transmission:\ it is equal to one if $\rho_{0}$ is orthogonal to $\rho_{1}$
and equal to zero if $\rho_{0}=\rho_{1}$. The quantum fidelity $F\left(
W\right)  $\ is another parameter that characterizes the quality of a cq
channel:%
\begin{equation}
F\left(  W\right)  \equiv F\left(  \rho_{0},\rho_{1}\right)  \equiv\left\Vert
\sqrt{\rho_{0}}\sqrt{\rho_{1}}\right\Vert _{1},
\end{equation}
where the trace norm $\left\Vert A\right\Vert _{1}$ of an operator $A$ is
defined as $\left\Vert A\right\Vert _{1}\equiv\ $Tr$\{  \sqrt{A^{\dag}%
A}\}  $ \cite{U73,J94}.\footnote{Note that the quantum fidelity
sometimes is defined as $\left\Vert \sqrt{\rho_{0}}\sqrt{\rho_{1}}\right\Vert
_{1}^{2}$ in order for it to have the interpretation as a probability. We
choose to remove the square in this work (as is often done) in order for it to
reduce to the classical Bhattacharya parameter when the states are just probability distributions.}
The quantum fidelity $F\left(  W\right)  $ is equal to one if $\rho_{0}%
=\rho_{1}$ and equal to zero if $\rho_{0}$ is orthogonal to $\rho_{1}$. We
have the following relationships between the symmetric Holevo information and
the quantum fidelity:%
\begin{align}
I\left(  W\right)   &  \approx1\Leftrightarrow F\left(  W\right)  \approx0,\\
I\left(  W\right)   &  \approx0\Leftrightarrow F\left(  W\right)  \approx1,
\end{align}
which are made precise in Proposition~1 of Ref.~\cite{WG11}.

From any cq channel, it is possible to induce a purely classical channel
$p_{Y|X}\left(  y|x\right)  $\ by having the receiver perform a quantum
measurement at its output:%
\begin{equation}
p_{Y|X}\left(  y|x\right)  \equiv\text{Tr}\left\{  \Lambda_{y}\rho
_{x}\right\}  ,
\end{equation}
where $\Lambda\equiv\left\{  \Lambda_{y}\right\}  $ is a positive
operator-valued measure (POVM), a set of operators satisfying $\Lambda_{y}%
\geq0$ and $\sum_{y}\Lambda_{y}=I$. Letting $X$ be a uniform Bernoulli random
variable and letting $Y$ be the random variable corresponding to the outcome
of the measurement, we can define the symmetric mutual information of the
induced channel as%
\begin{equation}
I\left(  W,\Lambda\right)  \equiv I\left(  X;Y\right)  \equiv H\left(
X\right)  +H\left(  Y\right)  -H\left(  XY\right)  ,
\end{equation}
where $H$ is the Shannon entropy of these random variables. The classical
Bhattarcharya parameter is the statistical overlap between the resulting
distributions:%
\begin{equation}
Z\left(  W,\Lambda\right)  \equiv\sum_{y}\sqrt{p_{Y|X}\left(  y|0\right)\
p_{Y|X}\left(  y|1\right)  }.
\end{equation}
If one were to encode the conditional distribution $p_{Y|X}\left(  y|x\right)
$ along the diagonal of a matrix (so that it becomes a density operator), then
it is clear that the symmetric Holevo information and fidelity of the
resulting \textquotedblleft cq channel\textquotedblright\ are equal to the
symmetric mutual information and classical Bhattacharya parameter, respectively.

The symmetric accessible information is equal to the optimized symmetric
mutual information:%
\begin{equation}
I\left(  W_{\text{acc}}\right)  \equiv\max_{\left\{  \Lambda_{y}\right\}
}I\left(  W,\Lambda\right)  ,
\end{equation}
where the optimization is with respect to all POVMs $\Lambda=\left\{
\Lambda_{y}\right\}  $. As a consequence of the well-known Holevo bound, 
the symmetric Holevo information
is an upper bound to the symmetric accessible information \cite{Holevo73}:%
\begin{equation}
I\left(  W_{\text{acc}}\right)  \leq I\left(  W\right)  .
\end{equation}

\section{The Fuchs-Caves Measurement}

\label{sec:Fuchs-Caves}

Rather than choosing a measurement to optimize the symmetric mutual
information, one could also choose a measurement in such a way that it
minimizes the statistical overlap between the resulting distributions
$p_{Y|X}\left(  y|0\right)  $ and $p_{Y|X}\left(  y|1\right)  $ \cite{FC95}.
We call such a measurement a \textquotedblleft Fuchs-Caves\textquotedblright%
\ measurement since these authors proved that the minimum statistical overlap
is equal to the quantum fidelity:%
\begin{equation}
\min_{\left\{  \Lambda_{y}\right\}  }Z\left(  W,\Lambda\right)  =F\left(
W\right)  .\label{eq:min-statistical-overlap}%
\end{equation}
Furthermore, they gave an explicit form for the measurement that achieves the
minimum and interpreted it as a kind of \textquotedblleft quantum likelihood
ratio.\textquotedblright\ Indeed, the measurement that achieves the minimum in
(\ref{eq:min-statistical-overlap}) corresponds to a measurement in the
eigenbasis of the following Hermitian operator:%
\begin{equation}
\rho_{0}\ \#\ \rho_{1}^{-1}\equiv\rho_{1}^{-1/2}\sqrt{\rho_{1}^{1/2}\rho
_{0}\rho_{1}^{1/2}}\rho_{1}^{-1/2}.
\end{equation}
Diagonalizing $\rho_{0}\ \#\ \rho_{1}^{-1}$ as%
\begin{equation}
\rho_{0}\ \#\ \rho_{1}^{-1}=\sum_{y}\lambda_{y}\left\vert y\right\rangle
\left\langle y\right\vert ,\label{eq:q-likelihood-decomp}%
\end{equation}
Fuchs and Caves observed that the eigenvalues of
$\rho_{0}\ \#\ \rho_{1}^{-1}$ take the following form:%
\begin{equation}
\lambda_{y}=\Bigg(  \frac{\left\langle y\right\vert \rho_{0}\left\vert
y\right\rangle }{\left\langle y\right\vert \rho_{1}\left\vert y\right\rangle
}\Bigg)  ^{1/2},
\end{equation}
furthermore suggesting that this measurement is a good quantum analog of a
likelihood ratio. In addition, Fuchs and Caves also observed that the operator
\begin{equation}
\rho_{1}\ \#\ \rho_{0}^{-1}\equiv\rho_{0}^{-1/2}\sqrt{\rho_{0}^{1/2}\rho
_{1}\rho_{0}^{1/2}}\rho_{0}^{-1/2}
\end{equation}
commutes with and is the inverse of $\rho_{0}\ \#\ \rho_{1}^{-1}$. Thus,
the eigenvectors of $\rho_{1}\ \#\ \rho_{0}^{-1}$ are the same as those of
$\rho_{0}\ \#\ \rho_{1}^{-1}$ and its eigenvalues are the reciprocals of those
of $\rho_{0}\ \#\ \rho_{1}^{-1}$.

\begin{lemma}
\label{lem:FC-error}When using the Fuchs-Caves measurement to distinguish
$\rho_{0}$ from $\rho_{1}$, we have following upper bound on the probability
of error $p_{e}\left(  W\right)  $\ in terms of the quantum fidelity~$F\left(
W\right)  $:%
\begin{equation}
p_{e}\left(  W\right)  \leq\tfrac{1}{2}F\left(  W\right)  .
\label{eq:err-prob-upper-bound}%
\end{equation}

\end{lemma}

\begin{proof}
After performing the measurement specified by (\ref{eq:q-likelihood-decomp}),
the decision rule is as follows:%
\begin{align}
\text{decide~}\rho_{0}\text{ if }\lambda_{y} &  \geq1,\\
\text{decide~}\rho_{1}\text{ if }\lambda_{y} &  <1,
\end{align}
which corresponds to the projectors%
\begin{align}
\Pi_{0} &  \equiv\sum_{y\ :\ \lambda_{y}\geq1}\left\vert y\right\rangle
\left\langle y\right\vert ,\label{eq:FC-project0}\\
\Pi_{1} &  =\sum_{y\ :\ \lambda_{y}<1}\left\vert y\right\rangle \left\langle
y\right\vert .\label{eq:FC-project1}%
\end{align}
It is then easy to prove the bound in (\ref{eq:err-prob-upper-bound}):%
\begin{align}
2\ p_{e}\left(  W\right)   &  =\text{Tr}\left\{  \Pi_{0}\rho_{1}\right\}
+\text{Tr}\left\{  \Pi_{1}\rho_{0}\right\}  \label{eq:FC-analysis-first}\\
&  =\sum_{y\ :\ \lambda_{y}\geq1}\left\langle y\right\vert \rho_{1}\left\vert
y\right\rangle +\sum_{y\ :\ \lambda_{y}<1}\left\langle y\right\vert \rho
_{0}\left\vert y\right\rangle \\
&  =\sum_{y\ :\ \lambda_{y}\geq1}\left\langle y\right\vert \rho_{1}\left\vert
y\right\rangle ^{1/2}\left\langle y\right\vert \rho_{1}\left\vert
y\right\rangle ^{1/2}+\sum_{y\ :\ \lambda_{y}<1}\left\langle y\right\vert
\rho_{0}\left\vert y\right\rangle ^{1/2}\left\langle y\right\vert \rho
_{0}\left\vert y\right\rangle ^{1/2}\\
&  \leq\sum_{y\ :\ \lambda_{y}\geq1}\left\langle y\right\vert \rho
_{1}\left\vert y\right\rangle ^{1/2}\left\langle y\right\vert \rho
_{0}\left\vert y\right\rangle ^{1/2}+\sum_{y\ :\ \lambda_{y}<1}\left\langle
y\right\vert \rho_{0}\left\vert y\right\rangle ^{1/2}\left\langle y\right\vert
\rho_{1}\left\vert y\right\rangle ^{1/2}\\
&  =\sum_{y}\left\langle y\right\vert \rho_{1}\left\vert y\right\rangle
^{1/2}\left\langle y\right\vert \rho_{0}\left\vert y\right\rangle ^{1/2}\\
&  =F\left(  \rho_{0},\rho_{1}\right)  \label{eq:FC-analysis-last}%
\end{align}
where the last equality follows from (\ref{eq:min-statistical-overlap}).
\end{proof}

\section{Decoding Classical-Quantum Polar Codes}

\subsection{Review}

\label{sec:polar-code-review}

Ref.~\cite{WG11} demonstrated how to construct synthesized versions of $W$, by
channel combining and splitting \cite{A09}. The synthesized channels
$W_{N}^{\left(  i\right)  }$ are of the following form:%
\begin{align}
W_{N}^{\left(  i\right)  } &  :\,u_{i}\rightarrow\rho_{\left(  i\right)
,u_{i}}^{U_{1}^{i-1}B^{N}},\label{eq:split-channels}\\
\rho_{\left(  i\right)  ,u_{i}}^{U_{1}^{i-1}B^{N}} &  \equiv\sum_{u_{1}^{i-1}%
}\frac{1}{2^{i-1}}\left\vert u_{1}^{i-1}\right\rangle \left\langle u_{1}%
^{i-1}\right\vert ^{U_{1}^{i-1}}\otimes\overline{\rho}_{u_{1}^{i}}^{B^{N}},\\
\overline{\rho}_{u_{1}^{i}}^{B^{N}} &  \equiv\sum_{u_{i+1}^{N}}\frac
{1}{2^{N-i}}\rho_{u^{N}G_{N}}^{B^{N}},\,\,\,\,\,\,\,\,\rho_{x^{N}}^{B^{N}%
}\equiv\rho_{x_{1}}^{B_{1}}\otimes\cdots\otimes\rho_{x_{N}}^{B_{N}},
\end{align}
where $G_{N}$ is Arikan's encoding circuit matrix built from classical
CNOT\ and permutation gates. The registers labeled by $U$ are classical registers
containing the bits $u_1$ through $u_{i-1}$, and the
 registers labeled by $B$ contain the channel outputs.
 If the channel is classical, then these states
are diagonal in the computational basis, and the above states correspond to
the distributions for the synthesized channels \cite{A09}. The interpretation
of $W_{N}^{\left(  i\right)  }$ is that it is the channel \textquotedblleft
seen\textquotedblright\ by the input $u_{i}$ if the previous bits $u_{1}%
^{i-1}$ are available and if the future bits $u_{i+1}^{N}$ are randomized.
This motivates the development of a quantum successive cancellation decoder
\cite{WG11}\ that attempts to distinguish $u_{i}=0$ from $u_{i}=1$ by
adaptively exploiting the results of previous measurements and quantum
hypothesis tests for each bit decision.

The synthesized channels $W_{N}^{\left(  i\right)  }$ polarize, in the sense
that some become nearly perfect for classical data transmission while others
become nearly useless. To prove this result, one can model the channel
splitting and combining process as a random birth process \cite{A09,WG11}, and
then demonstrate that the induced random birth processes corresponding to
the channel parameters $I(W_{N}^{\left(  i\right)  })$ and $F(W_{N}^{\left(
i\right)  })$ are martingales that converge almost surely to zero-one valued
random variables in the limit of many recursions. The following theorem
characterizes the rate with which the channel polarization effect takes hold
\cite{AT09,WG11}, and it is useful in proving statements about the performance
of polar codes for cq channels:

\begin{theorem}
\label{thm:fraction-good}Given a binary input cq channel $W$ and any
$\beta<1/2$, it holds that%
\begin{equation}
\lim_{n\rightarrow\infty}\Pr_{J}\{F(W_{2^{n}}^{\left(  J\right)
})<2^{-2^{n\beta}}\}=I\left(  W\right)  ,
\end{equation}
where $n$ indicates the level of recursion for the encoding, $W_{2^{n}%
}^{\left(  J\right)  }$ is a random variable characterizing the $J^{\text{th}%
}$ split channel, and $F(W_{2^{n}}^{\left(  J\right)  })$ is the fidelity of
that channel.
\end{theorem}

Assuming knowledge of the identities of the good and bad channels, one can then construct a
coding scheme based on the channel polarization effect, by dividing the
synthesized channels according to the following polar coding rule:%
\begin{align}
\mathcal{G}_{N}\left(  W,\beta\right)   &  \equiv\big\{i\in\left[  N\right]
:F(W_{N}^{\left(  i\right)  })<2^{-N^{\beta}}%
\big\},\label{eq:polar-coding-rule}\\
\mathcal{B}_{N}\left(  W,\beta\right)   &  \equiv\left[  N\right]
\setminus\mathcal{G}_{N}\left(  W,\beta\right)  ,
\end{align}
so that $\mathcal{G}_{N}\left(  W,\beta\right)  $ is the set of
\textquotedblleft good\textquotedblright\ channels and $\mathcal{B}_{N}\left(
W,\beta\right)  $ is the set of \textquotedblleft bad\textquotedblright%
\ channels. The sender then transmits the information bits through the good
channels and \textquotedblleft frozen\textquotedblright\ bits through the bad
ones. A helpful assumption for error analysis is that the frozen bits are
chosen uniformly at random and known to both the sender and receiver.

One of the important advances in Ref.~\cite{WG11} was to establish that a
quantum successive cancellation decoder performs well for polar coding over
classical-quantum channels with equiprobable inputs. Corresponding to the
split channels $W_{N}^{\left(  i\right)  }$ in (\ref{eq:split-channels}) are
the following projectors that attempt to decide whether the input of the
$i^{\text{th}}$ split channel is zero or one:%
\begin{align}
\Pi_{\left(  i\right)  ,0}^{U_{1}^{i-1}B^{N}} &  \equiv\left\{  \rho_{\left(
i\right)  ,0}^{U_{1}^{i-1}B^{N}}-\rho_{\left(  i\right)  ,1}^{U_{1}^{i-1}%
B^{N}}\geq0\right\}  ,\label{eq:Helstrom-decoder-1}\\
\Pi_{\left(  i\right)  ,1}^{U_{1}^{i-1}B^{N}} &  \equiv I-\Pi_{\left(
i\right)  ,0}^{U_{1}^{i-1}B^{N}},\label{eq:Helstrom-decoder-2}%
\end{align}
where $\left\{  B\geq0\right\}  $ denotes the projector onto the positive
eigenspace of a Hermitian operator~$B$. After some calculations, one
readily sees that%
\begin{equation}
\Pi_{\left(  i\right)  ,0}^{U_{1}^{i-1}B^{N}}=\sum_{u_{1}^{i-1}}\left\vert
u_{1}^{i-1}\right\rangle \left\langle u_{1}^{i-1}\right\vert ^{U_{1}^{i-1}%
}\otimes\Pi_{\left(  i\right)  ,u_{1}^{i-1}0}^{B^{N}}%
,\label{eq:projectors-expanded-1}%
\end{equation}
where
\begin{align}
\Pi_{\left(  i\right)  ,1}^{U_{1}^{i-1}B^{N}} &  =I-\Pi_{\left(  i\right)
,0}^{U_{1}^{i-1}B^{N}},\label{eq:projectors-expanded-2}\\
\Pi_{\left(  i\right)  ,u_{1}^{i-1}0}^{B^{N}} &  \equiv\{\overline{\rho
}_{u_{1}^{i-1}0}^{B^{N}}-\overline{\rho}_{u_{1}^{i-1}1}^{B^{N}}\geq0\},\\
\Pi_{\left(  i\right)  ,u_{1}^{i-1}1}^{B^{N}} &  \equiv I-\Pi_{\left(
i\right)  ,u_{1}^{i-1}0}^{B^{N}}.\label{eq:projectors-expanded-3}%
\end{align}
The observations above lead to a decoding rule for a successive cancellation
decoder similar to Arikan's~\cite{A09}:%
\begin{equation}
\hat{u}_{i}=\left\{
\begin{array}
[c]{cc}%
u_{i} & \text{if }i\in\mathcal{A}^{c}\\
h\left(  \hat{u}_{1}^{i-1}\right)   & \text{if }i\in\mathcal{A}%
\end{array}
\right.  ,
\end{equation}
where $h\left(  \hat{u}_{1}^{i-1}\right)  $ is the outcome of the 
$i^{\text{th}}$ collective measurement:
\begin{equation}
\{\Pi_{\left(  i\right)  ,\hat{u}_{1}^{i-1}0}^{B^{N}},\,\Pi_{\left(  i\right)
,\hat{u}_{1}^{i-1}1}^{B^{N}}\}
\end{equation}
on the codeword received at the
channel output (after $i-1$ measurements have already been performed). The 
set $\mathcal{A}$ labels the information bits.
The measurement device outputs \textquotedblleft%
0\textquotedblright\ if the outcome $\Pi_{\left(  i\right)  ,\hat{u}_{1}%
^{i-1}0}^{B^{N}}$ occurs and it outputs \textquotedblleft1\textquotedblright%
\ otherwise. (Note that we can set $\Pi_{\left(  i\right)  ,\hat{u}_{1}%
^{i-1}u_{i}}^{B^{N}}=I$ if the bit $u_{i}$ is a frozen bit.) The above
sequence of measurements for the whole bit stream $u^{N}$ corresponds to a
positive operator-valued measure (POVM)~$\left\{  \Lambda_{u^{N}}\right\}  $
where%
\begin{equation}
\Lambda_{u^{N}}\equiv\Pi_{\left(  1\right)  ,u_{1}}^{B^{N}}\cdots\Pi_{\left(
i\right)  ,u_{1}^{i-1}u_{i}}^{B^{N}}\cdots\Pi_{\left(  N\right)  ,u_{1}%
^{N-1}u_{N}}^{B^{N}}\cdots\Pi_{\left(  i\right)  ,u_{1}^{i-1}u_{i}}^{B^{N}%
}\cdots\Pi_{\left(  1\right)  ,u_{1}}^{B^{N}},
\end{equation}
and $\sum_{u_{\mathcal{A}}}\Lambda_{u^{N}}=I^{B^{N}}$. The probability of
error $P_{e}\left(  N,K,\mathcal{A},u_{\mathcal{A}^{c}}\right)  $\ for code
length $N$, number~$K$ of information bits, set $\mathcal{A}$ of information
bits, and choice $u_{\mathcal{A}^{c}}$ for the frozen bits is%
\begin{equation} \label{eqn:avg.error}
P_{e}\left(  N,K,\mathcal{A},u_{\mathcal{A}^{c}}\right)  =1-\frac{1}{2^{K}%
}\sum_{u_{\mathcal{A}}}\text{Tr}\left\{  \Lambda_{u^{N}}\rho_{u^{N}}\right\}
,
\end{equation}
where we are assuming a particular choice of the bits $u_{\mathcal{A}^{c}}$ in
the sequence of projectors $\Pi_{\left(  N\right)  ,u_{1}^{N-1}u_{N}}^{B^{N}}$
$\cdots$ $\Pi_{\left(  i\right)  ,u_{1}^{i-1}u_{i}}^{B^{N}}$ $\cdots$
$\Pi_{\left(  1\right)  ,u_{1}}^{B^{N}}$ and setting $\Pi_{\left(  i\right)
,u_{1}^{i-1}u_{i}}^{B^{N}}=I$ if $u_{i}$ is a frozen bit. The formula also assumes 
that  the sender transmits the information 
sequence $u_{\mathcal{A}}$ with
uniform probability $2^{-K}$. The probability of error averaged over all
choices of the frozen bits is then
\begin{equation}
P_{e}\left(  N,K,\mathcal{A}\right)  =\frac{1}{2^{N-K}}\sum_{u_{\mathcal{A}%
^{c}}}P_{e}\left(  N,K,\mathcal{A},u_{\mathcal{A}^{c}}\right) .
\end{equation}
The following proposition from Ref.~\cite{WG11} determines an upper bound on
the average ensemble performance of polar codes with a quantum successive
cancellation decoder:

\begin{proposition}
\label{prop:error-bound}For any classical-quantum channel $W$ with binary
inputs and quantum outputs and any choice of $\left(  N,K,\mathcal{A}\right)
$, the following bound holds%
\begin{equation}
P_{e}\left(  N,K,\mathcal{A}\right)  \leq2\sqrt{\sum_{i\in\mathcal{A}}\tfrac
{1}{2}F(W_{N}^{\left(  i\right)  })}. \label{eq:prop-error-bound}%
\end{equation}

\end{proposition}

The proposition is proved by exploiting Sen's non-commutative
union bound~\cite{S11} and Lemma 3.2 of Ref.~\cite{H06}
(which upper bounds the probability of error in a binary quantum hypothesis
test by the fidelity between the test states). 
%
% For completeness, we state
% Sen's non-commutative union bound:%
% \begin{equation}
% \text{Tr}\left\{  \sigma\right\}  -\text{Tr}\left\{  P_{N}\cdots P_{1}\sigma
% P_{1}\cdots P_{N}\right\}  \leq2\sqrt{\sum_{i=1}^{N}\text{Tr}\left\{  \left(
% I-P_{i}\right)  \sigma\right\}  },
% \end{equation}
% which holds for any subnormalized state $\sigma$ (with $\sigma\geq0$ and
% Tr$\left\{  \sigma\right\}  \leq1$) and projectors $P_{1}$, \ldots, $P_{N}$.
 The bound in (\ref{eq:prop-error-bound}) applies provided the
 sender chooses the information bits $U_{\mathcal{A}}$ from a uniform
 distribution. Thus, by choosing the channels over which the sender transmits
 the information bits to be in $\mathcal{A}$ and those over which she transmits
 agreed-upon frozen bits to be in $\mathcal{A}^{c}$, we obtain that the
 probability of decoding error satisfies
 $\Pr\{\widehat{U}_{\mathcal{A}}\neq U_{\mathcal{A}}\}=o(2^{-\frac{1}{2}%
 N^{\beta}})$,
 as long as the code rate obeys
 $R=K/N<I(W)$.
%This completes the specification of a cq polar code.

A final point that will be useful is that Ref.~\cite{WG11} also proved that measurements consisting
of the projections%
\begin{equation}
\left\{  \sqrt{\rho_{\left(  i\right)  ,0}^{U_{1}^{i-1}B^{N}}}-\sqrt
{\rho_{\left(  i\right)  ,1}^{U_{1}^{i-1}B^{N}}}\geq0\right\}  ,
\label{eq:sq-root-Helstrom}%
\end{equation}
rather than those in (\ref{eq:Helstrom-decoder-1}%
)-(\ref{eq:Helstrom-decoder-2}), also achieve the performance stated in
Proposition~\ref{prop:error-bound}.

\subsection{Collective Fuchs-Caves Measurements Achieve the Holevo Rate}

\label{sec:baby-result}

Our first observation is rather simple, just being that collective Fuchs-Caves
measurements can also achieve the performance stated in
Proposition~\ref{prop:error-bound}. This result follows from
Lemma~\ref{lem:FC-error}'s bound on the
error probability of a Fuchs-Caves measurement 
 and by performing an error analysis similar to that in the proof of
Proposition~4 of Ref.~\cite{WG11}\ given in Section~V of that paper. The
explicit form of a Fuchs-Caves quantum successive cancellation decoder is
given by projectors of the form in (\ref{eq:projectors-expanded-1}%
)-(\ref{eq:projectors-expanded-3}), with the Helstrom tests replaced by
Fuchs-Caves projectors as given in (\ref{eq:FC-project0}%
)-(\ref{eq:FC-project1}).

This result also demonstrates that there are a variety of decoding measurements
that one can exploit for achieving the Holevo information rate. However, the
quantum successive cancellation decoder consisting of Helstrom measurements
should outperform either the measurements in (\ref{eq:sq-root-Helstrom}) or
the Fuchs-Caves measurements when considering finite blocklength performance
because the Helstrom measurement is the optimal test for distinguishing two
quantum states.

\subsection{Main Result}

\label{sec:main-result}

Our main observation is a bit more subtle than the above, but it is still
elementary. Nevertheless, this observation has nontrivial consequences and
represents a step beyond the insights in prior work regarding decoding of
classical information sent over quantum channels
\cite{Hol98,SW97,GGLMSY04,G11,GLM10,S11,WGTS11,WG11}.

We begin by considering the \textquotedblleft Fuchs-Caves\textquotedblright%
\ classical channel $W_{\text{FC}}$\ induced from $W$ by performing the
Fuchs-Caves measurement on every channel output:%
\begin{equation}
W_{\text{FC}}:x\rightarrow p_{Y|X}\left(  y|x\right)  =\left\langle
y\right\vert \rho_{x}\left\vert y\right\rangle ,
\end{equation}
where the orthonormal basis $\left\{  \left\vert y\right\rangle \right\}  $ is
the same as that in (\ref{eq:q-likelihood-decomp}). The specification of the polar
code in the previous section specializes to this induced classical channel.
The code consists of a set of
\textquotedblleft good\textquotedblright\ synthesized channels $\mathcal{G}%
_{N}\left(  W_{\text{FC}},\beta\right)  $ and \textquotedblleft
bad\textquotedblright\ synthesized channels $\mathcal{B}_{N}\left(
W_{\text{FC}},\beta\right)  $, where%
\begin{align}
\mathcal{G}_{N}\left(  W_{\text{FC}},\beta\right)   &  \equiv\big\{i\in\left[
N\right]  :F(W_{\text{FC},N}^{\left(  i\right)  })=Z(W_{\text{FC},N}^{\left(
i\right)  })<2^{-N^{\beta}}\big\},\\
\mathcal{B}_{N}\left(  W_{\text{FC}},\beta\right)   &  \equiv\left[  N\right]
\setminus\mathcal{G}_{N}\left(  W_{\text{FC}},\beta\right)  ,
\end{align}
and the equality $F(W_{\text{FC},N}^{\left(  i\right)  })=Z(W_{\text{FC}%
,N}^{\left(  i\right)  })$ holds because the induced channels are classical.
Furthermore, by Theorem~\ref{thm:fraction-good}, the number of good channels
in the limit that $N$ becomes large is as follows:%
\begin{equation}
\lim_{N\rightarrow\infty}\frac{1}{N}\left\vert \mathcal{G}_{N}\left(
W_{\text{FC}},\beta\right)  \right\vert =I\left(  W_{\text{FC}}\right)
.\label{eq:rate-FC}%
\end{equation}
Finally, each bit of this classical polar code can be decoded in time
$O\left(  N\right)  $ using a recursive calculation of likelihood ratios as
given in (75)-(76) of Ref.~\cite{A09}.\footnote{Note that this is the
\textquotedblleft first decoding algorithm\textquotedblright\ of Arikan. A
refinement implies that all of the bits can be decoded in time $O\left(  N\log
N\right)  $, but the first decoding algorithm is what we will use in this
work.}

Now, our main observation is the following relationship between the good
channels of~$W_{\text{FC}}$ and the good channels of $W$:%
\begin{equation}
\mathcal{G}_{N}\left(  W_{\text{FC}},\beta\right)  \subseteq\mathcal{G}%
_{N}\left(  W,\beta\right)  . \label{eq:main-obs-1}%
\end{equation}
This relationship holds because of the Fuchs-Caves formula from
(\ref{eq:min-statistical-overlap}). For all $i$, we have that%
\begin{equation}
F(W_{N}^{\left(  i\right)  })=\min_{\left\{  \Lambda_{y}\right\}  }Z(
W_{N}^{\left(  i\right)  },\Lambda)  \leq Z(  W_{\text{FC}%
,N}^{\left(  i\right)  })  , \label{eq:main-observation}%
\end{equation}
where the inequality follows because the tensor-product Fuchs-Caves
measurement that induces the synthesized channel $W_{\text{FC},N}^{\left(
i\right)  }$ is a particular kind of measurement, and so its classical
statistical overlap can only be larger than that realized by the optimal
measurement (which in general will be a collective measurement rather than a
product measurement). Now, for all $i\in\mathcal{G}_{N}\left(  W_{\text{FC}%
},\beta\right)  $, we have that%
\begin{equation}
Z(W_{\text{FC},N}^{\left(  i\right)  })<2^{-N^{\beta}}.
\label{eq:Fuchs-Caves-good-errors}%
\end{equation}
This in turn implies that $F(W_{N}^{\left(  i\right)  })<2^{-N^{\beta}}$ by
(\ref{eq:main-observation}), and so for this $i$, we have that $i\in
\mathcal{G}_{N}\left(  W,\beta\right)  $ and can conclude (\ref{eq:main-obs-1}).

This observation has non-trivial implications for the structure of the polar
decoder. For all of the bits in $\mathcal{G}_{N}\left(  W_{\text{FC}}%
,\beta\right)  $, the receiver can decode them with what amounts to an
effectively \textquotedblleft product\textquotedblright\ or \textquotedblleft
non-collective\textquotedblright\ strategy,\footnote{If a decoding strategy
amounts to coherent implementations of product measurements followed by
coherent processing of the outcomes, we still say that it is a product
strategy rather than collective.} while for the bits in $\mathcal{G}%
_{N}\left(  W,\beta\right)  \setminus\mathcal{G}_{N}\left(  W_{\text{FC}%
},\beta\right)  $, we still require collective measurements in order for the
receiver to decode them with the error probability guarantee given by
(\ref{eq:polar-coding-rule}). However, when decoding the bits in $\mathcal{G}%
_{N}\left(  W_{\text{FC}},\beta\right)  $, the receiver should be careful to
decode them in the least destructive way possible so that Sen's
non-commutative union bound is still applicable and we obtain the overall
error bound guaranteed by Proposition~\ref{prop:error-bound}. In particular,
the decoder should begin by performing an isometric extension of the
Fuchs-Caves measurement on each channel output:%
\begin{equation}
\sum_{y}\left\vert y\right\rangle \left\langle y\right\vert \otimes\left\vert
\lambda_{y}\right\rangle ,\label{eq:FC-isometry}%
\end{equation}
where the orthonormal basis $\left\{  \left\vert y\right\rangle \right\}  $ is
from the eigendecomposition in (\ref{eq:q-likelihood-decomp}) and the basis
$\left\{  \left\vert \lambda_{y}\right\rangle \right\}  $ encodes the
eigenvalues to some finite precision. Such an operation coherently copies the
likelihood ratios $\lambda_{y}$ of the Fuchs-Caves measurement into an
ancillary register. The receiver then performs a reversible implementation of
Arikan's decoding algorithm to process these likelihood ratios according to
(75)-(76) of Ref.~\cite{A09}. Finally, the receiver coherently copies the
value of a single decision qubit with a CNOT gate to an ancillary register,
measures the decision qubit, and \textquotedblleft
uncomputes\textquotedblright\ these operations by performing the inverse of
the Arikan circuit and the inverse of the operations in (\ref{eq:FC-isometry}%
). Figure~\ref{fig:FC-coherent-decoder}\ depicts these operations. The effect
of these operations is to implement a projection of the channel output onto a
subspace spanned by eigenvectors $|y^{N}\rangle=\left\vert y_{1}\right\rangle
\otimes\cdots\otimes\left\vert y_{N}\right\rangle $ of the Fuchs-Caves
measurements such that%
\begin{equation}
W_{\text{FC},N}^{\left(  i\right)  }\left(  y^{N},u_{1}^{i-1}|0\right)  \geq
W_{\text{FC},N}^{\left(  i\right)  }\left(  y^{N},u_{1}^{i-1}|1\right)  ,
\end{equation}
or onto the complementary subspace spanned by eigenvectors $|y^{N}\rangle$
such that%
\begin{equation}
W_{\text{FC},N}^{\left(  i\right)  }\left(  y^{N},u_{1}^{i-1}|0\right)
<W_{\text{FC},N}^{\left(  i\right)  }\left(  y^{N},u_{1}^{i-1}|1\right)  ,
\end{equation}
where $y^{N}$ is the classical output of the Fuchs-Caves channel and
$u_{1}^{i-1}$ denotes the previously decoded bits. Thus, the fidelity bound
from (\ref{eq:Fuchs-Caves-good-errors}) is applicable and Sen's
non-commutative union bound guarantees that the overall contribution of the
error in decoding bit $i\in\mathcal{G}_{N}\left(  W_{\text{FC}},\beta\right)
$ is no larger than $2^{-N^{\beta}}$. The time that it takes to process each
bit $i\in\mathcal{G}_{N}\left(  W_{\text{FC}},\beta\right)  $ is $O\left(
N\right)  $, which is clear from the structure of the circuit and
Arikan's \textquotedblleft first decoding algorithm.\textquotedblright

\begin{figure}[ptb]
\begin{center}
\includegraphics[
width=5.0502in
]{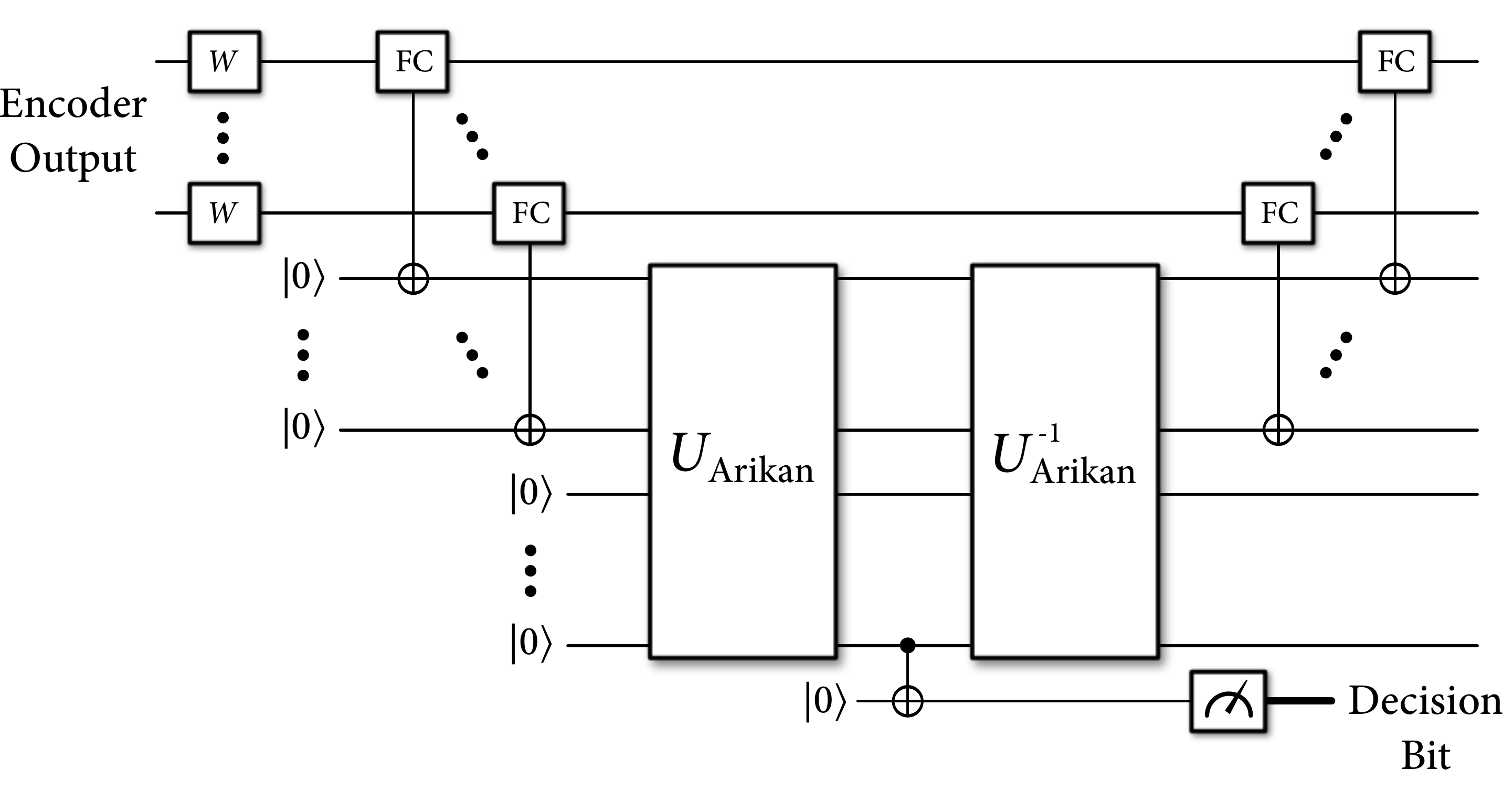}
\end{center}\vspace{-.30in}
\caption{The circuit for recovering an information bit in the set $\mathcal{G}%
_{N}\left(  W_{\text{FC}},\beta\right)  $. The  encoder output  is fed
into $N$ instances of\ the channel $W$. The receiver acts with $N$ of the
unitaries in (\ref{eq:FC-isometry}), labeled as \textquotedblleft
FC\textquotedblright\ boxes which coherently copy the likelihood ratios
$\lambda_{y_{1}}$, \ldots, $\lambda_{y_{N}}$ into ancillary registers. The
receiver then acts with a reversible implementation of Arikan's likelihood
ratio computations, copies the decision bit into an ancillary register, and
measures the decision bit to decode the $i^{\text{th}}$ bit. The receiver
finally performs the inverse of these operations to \textquotedblleft clean
up,\textquotedblright\ i.e., to ensure that the next measurement can be
performed, whether it be to decode a bit in the set $\mathcal{G}_{N}\left(
W_{\text{FC}},\beta\right)  $ or the set $\mathcal{G}_{N}\left(
W,\beta\right)  \setminus\mathcal{G}_{N}\left(  W_{\text{FC}},\beta\right)  $.
The effect of this circuit is to perform the desired \textquotedblleft gentle
projection.\textquotedblright}%
\label{fig:FC-coherent-decoder}%
\end{figure}

For all of the remaining bits $i\in\mathcal{G}_{N}\left(  W,\beta\right)  \setminus
\mathcal{G}_{N}\left(  W_{\text{FC}},\beta\right)  $, we still do not know
whether there exists an efficient quantum algorithm for decoding them while
having the error probability from Proposition~\ref{prop:error-bound}. Thus,
for now, we simply suggest for the receiver to use collective measurements to
recover them.

It should be clear from Proposition~\ref{thm:fraction-good} and
(\ref{eq:rate-FC}) that the size of the set $\mathcal{G}_{N}\left(
W,\beta\right)  \setminus\mathcal{G}_{N}\left(  W_{\text{FC}},\beta\right)  $
in the limit is equal to%
\begin{equation}
\lim_{N\rightarrow\infty}\frac{1}{N}\left\vert \mathcal{G}_{N}\left(
W,\beta\right)  \setminus\mathcal{G}_{N}\left(  W_{\text{FC}},\beta\right)
\right\vert =I\left(  W\right)  -I\left(  W_{\text{FC}}\right)  .
\end{equation}
This makes it clear that one does not require a collective strategy in order
to recover all of the information bits, but a collective strategy is only
required in order to bridge the gap between $I\left(  W_{\text{FC}}\right)  $
and $I\left(  W\right)  $.

Observe also that similar reasoning applies to any product measurement,
not just the Fuchs-Caves measurements (we
focused on the Fuchs-Caves measurement due to its strong analogy with a
likelihood ratio and because Arikan's decoding algorithm processes likelihood
ratios). With this in mind, we could simply choose the product measurement to
be the one that maximizes the accessible information, in order to maximize the
number of bits that can be processed efficiently. Let $W_{\text{acc}}$ be the
classical channel induced by performing the measurement that maximizes the
accessible information. One would then process the bits in $\mathcal{G}%
_{N}\left(  W_{\text{acc}},\beta\right)  $ in a way very similar as described
above. All of the observations above then justify Claim~\ref{clm:main-result}.

The reasoning also leads to a generalization of
Lemma~\ref{lem:FC-error}\ that applies when using Fuchs-Caves measurements to
distinguish a tensor-product state $\rho_{0}^{\otimes N}$ from $\rho
_{1}^{\otimes N}$. The test consists of performing product measurements
followed by classical post-processing. If one wishes to perform this test in
the most delicate way possible, one could perform it as in
Figure~\ref{fig:FC-coherent-decoder}.

\begin{lemma}
\label{lem:FC-error-N-copies}
When using product Fuchs-Caves measurements to
distinguish $\rho_{0}^{\otimes N}$ from $\rho_{1}^{\otimes N}$, the probability of
error $p_e$ is bounded from above in terms of the quantum fidelity $F\left(  \rho_{0},\rho_{1}\right)$:
\begin{equation}
p_{e}\leq\frac{1}{2}\left[  F\left(  \rho_{0},\rho_{1}\right)  \right]  ^{N}.
\end{equation}

\end{lemma}

\begin{proof}
The proof is very similar to the proof of Lemma~\ref{lem:FC-error}. 
The test, though, consists of performing individual Fuchs-Caves measurements on the $N$
systems, and these tests result in likelihood ratios $\lambda_{y_{1}}$,
\ldots, $\lambda_{y_{N}}$. The decision rule is then as follows:%
\begin{align}
\text{decide~}\rho_{0}^{\otimes N}\text{ if }\lambda_{y_{1}}\times\cdots
\times\lambda_{y_{N}} &  \geq1,\label{eq:decision-rule-1}\\
\text{decide~}\rho_{1}^{\otimes N}\text{ if }\lambda_{y_{1}}\times\cdots
\times\lambda_{y_{N}} &  <1.\label{eq:decision-rule-2}%
\end{align}
An analysis proceeding exactly as in\ (\ref{eq:FC-analysis-first}%
)-(\ref{eq:FC-analysis-last}) leads to the following bound:%
\begin{align}
2\ p_{e}\left(  W\right)   &  \leq\sum_{y_{1},\ldots,y_{N}}\left[
\left\langle y_{1}\right\vert \cdots\left\langle y_{N}\right\vert \rho
_{1}^{\otimes N}\left\vert y_{1}\right\rangle \cdots\left\vert y_{N}%
\right\rangle \right]  ^{1/2}\ \left[  \left\langle y_{1}\right\vert
\cdots\left\langle y_{N}\right\vert \rho_{0}^{\otimes N}\left\vert
y_{1}\right\rangle \cdots\left\vert y_{N}\right\rangle \right]  ^{1/2} \nonumber\\
&  =\sum_{y_{1},\ldots,y_{N}}\left\langle y_{1}\right\vert \rho_{1}\left\vert
y_{1}\right\rangle ^{1/2}\ \cdots\ \left\langle y_{N}\right\vert \rho
_{1}\left\vert y_{N}\right\rangle ^{1/2}\ \left\langle y_{1}\right\vert
\rho_{0}\left\vert y_{1}\right\rangle ^{1/2}\ \cdots\ \left\langle
y_{N}\right\vert \rho_{0}\left\vert y_{N}\right\rangle ^{1/2}\\
&  =\sum_{y_{1}}\left\langle y_{1}\right\vert \rho_{1}\left\vert
y_{1}\right\rangle ^{1/2}\left\langle y_{1}\right\vert \rho_{0}\left\vert
y_{1}\right\rangle ^{1/2}\ \cdots\ \sum_{y_{N}}\left\langle y_{N}\right\vert
\rho_{1}\left\vert y_{N}\right\rangle ^{1/2}\left\langle y_{N}\right\vert
\rho_{0}\left\vert y_{N}\right\rangle ^{1/2}\\
&  =\left[  F\left(  \rho_{0},\rho_{1}\right)  \right]  ^{N}.
\end{align}
Furthermore, one can implement this test efficiently and non-destructively on
a quantum computer as described in Figure~\ref{fig:FC-coherent-decoder}. The
result is to project onto two different subspaces:\ the one spanned by
eigenvectors whose corresponding eigenvalues satisfy (\ref{eq:decision-rule-1}%
) and the other.% (\ref{eq:decision-rule-2}).
\end{proof}

\section{Decoding the Pure-Loss Bosonic Channel}

\label{sec:bosonic}
\begin{figure}[ptb]
\begin{center}
\includegraphics[
width=4.3502in
]{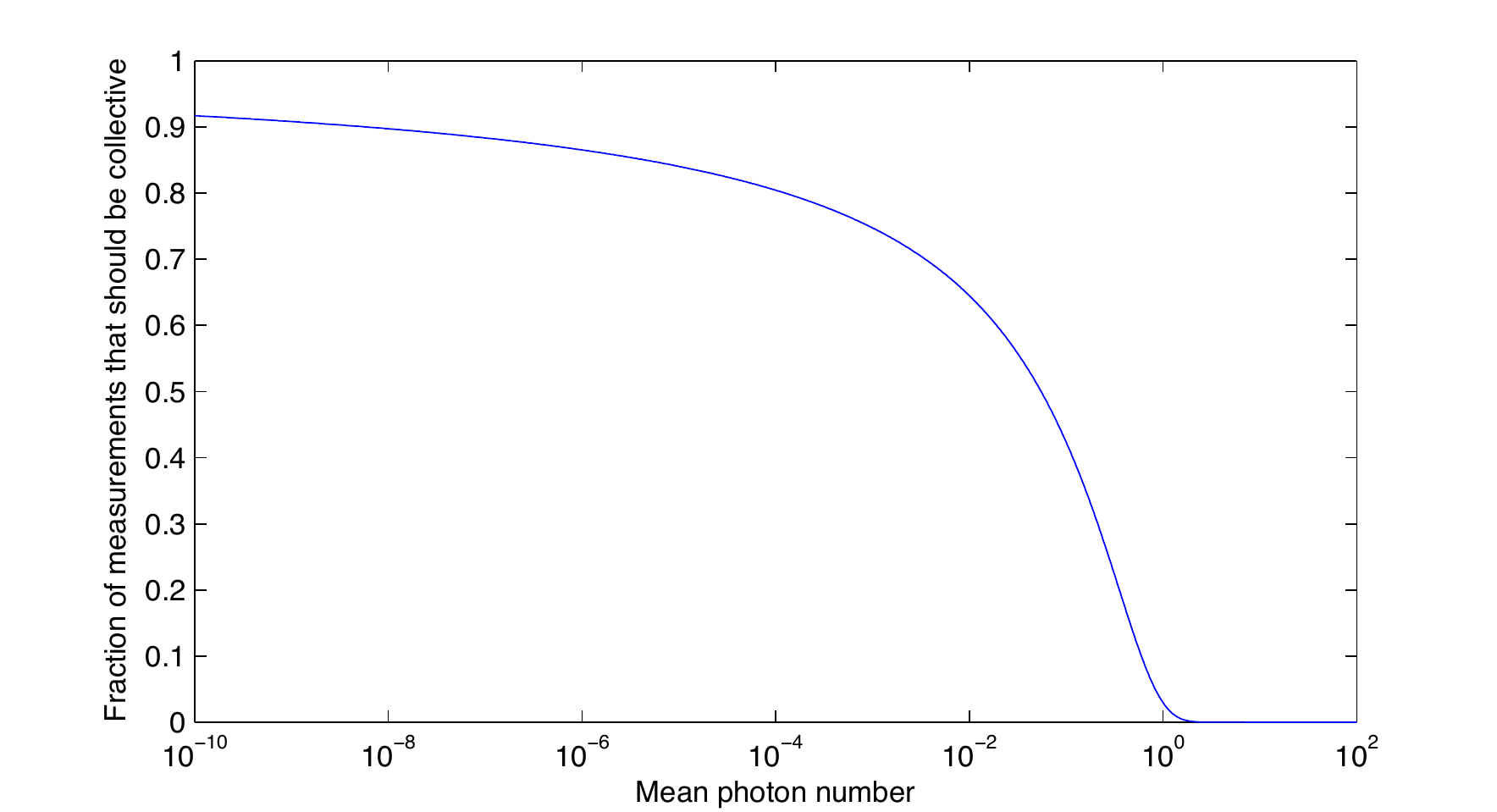}
\end{center}
\caption{The fraction of collective measurements required for a polar decoder
plotted as a function of the mean photon number $E$ at the receiving end, when
using a BPSK coding strategy.}%
\label{fig:fraction}%
\end{figure}
A channel of particular practical interest is the pure-loss bosonic channel. A
simple physical model for this channel is a beamsplitter of transmissivity
$\eta\in\left[  0,1\right]  $, where the sender has access to one input port,
the environment injects the vacuum state into the other input port, the
receiver has access to one output port, and the environment obtains the other
output port. It is well known that the Holevo capacity of this channel is
equal to $g\left(  \eta N_{S}\right)  \equiv\left(  \eta N_{S}+1\right)
\log\left(  \eta N_{S}+1\right)  -\eta N_{S}\log\left(  \eta N_{S}\right)  $
\cite{GGLMSY04}, where $N_{S}$ is the mean input photon number. In the
low-photon number regime, one can come very close to achieving the capacity by
employing a binary phase-shift keying (BPSK) strategy (using coherent states
$\left\vert \alpha\right\rangle $ and $\left\vert -\alpha\right\rangle $ as
the signaling states) \cite{MO00}. The BPSK strategy induces a cq channel of
the following form: $x\rightarrow\left\vert \left(  -1\right)  ^{x}%
\alpha\right\rangle \left\langle \left(  -1\right)  ^{x}\alpha\right\vert $.
The symmetric Holevo rate for this channel is equal to $\chi\left(  E\right)
\equiv h_{2}\left(  \left[  1+e^{-2E}\right]  /2\right)  $, where $h_{2}$ is
the binary entropy and $E\equiv\eta N_{S}$. If the receiver performs a
Helstrom measurement at every channel output, this induces a classical channel
with symmetric mutual information equal to $I_{\text{Hel}}\left(  E\right)
\equiv1-h_{2}(  [  1-\sqrt{1-e^{-4E}}]  /2)  $. (See
Ref.~\cite{GW12}, for example, for explicit calculations.) Our results in the
previous section demonstrate that the fraction of information bits required to
be decoded using a collective strategy is equal to $1-I_{\text{Hel}}\left(
E\right)  /\chi\left(  E\right)  $. Figure~\ref{fig:fraction} reveals that
this fraction is rather small for mean photon number (MPN) larger than one, but then
it rises sharply as we enter a quantum regime where the MPN is
less than one. Even deep in the quantum regime at a MPN of $10^{-8}$, however, roughly 10\% of the bits do not require collective decoding.

\section{Small Blocklength
Polar Decoders}

\label{sec:last-result}This section briefly discusses how the
Helstrom measurements \cite{H69,Hel76}\ in the quantum successive cancellation
decoder from Ref.~\cite{WG11}\ decompose for very small size polar codes.

\subsection{Two-Bit Polar Decoder}

We begin by considering the simple two-bit polar code. The channel is of the
form $x\rightarrow\rho_{x}$, where $x\in\left\{  0,1\right\}  $ and $\rho_{x}$
is some conditional density operator. The two-bit polar code performs the
simple transformation on the input bits $u_{1}$ and $u_{2}$:%
\begin{equation}
\left(  u_{1},u_{2}\right)  \rightarrow\left(  u_{1}+u_{2},u_{2}\right)  ,
\label{eq:first-polar-encoder}%
\end{equation}
where addition is modulo 2.

The first step of the successive cancellation decoder is to recover $u_{1}$,
assuming that bit~$u_{2}$ is chosen uniformly at random. The optimal
measurement is a Helstrom measurement, and in this case, it amounts to
distinguishing between the following two states%
\begin{equation}
\frac{1}{2}\sum_{u_{2}}\rho_{u_{2}}\otimes\rho_{u_{2}},\ \ \ \ \ \ \ \frac
{1}{2}\sum_{u_{2}}\rho_{u_{2}+1}\otimes\rho_{u_{2}}.
\end{equation}
The Helstrom measurement is given by the projector onto the positive
eigenspace of the difference of the two density operators above:%
\begin{align}
\left\{  \frac{1}{2}\sum_{u_{2}}\rho_{u_{2}}\otimes\rho_{u_{2}}-\frac{1}%
{2}\sum_{u_{2}}\rho_{u_{2}+1}\otimes\rho_{u_{2}}\geq0\right\}   &  =\left\{
\sum_{u_{2}}\left(  \rho_{u_{2}}-\rho_{u_{2}+1}\right)  \otimes\rho_{u_{2}%
}\geq0\right\}  \\
&  =\left\{  \sum_{u_{2}}\left(  -1\right)  ^{u_{2}}\left(  \rho_{0}-\rho
_{1}\right)  \otimes\rho_{u_{2}}\geq0\right\}  \\
&  =\left\{  \left(  \rho_{0}-\rho_{1}\right)  \otimes\sum_{u_{2}}\left(
-1\right)  ^{u_{2}}\rho_{u_{2}}\geq0\right\}  \\
&  =\left\{  \left(  \rho_{0}-\rho_{1}\right)  \otimes\left(  \rho_{0}%
-\rho_{1}\right)  \geq0\right\}  .\label{eq:1st-two-qubit-test}%
\end{align}
Thus, this test factorizes into the parity of the individual quantum hypothesis
tests $\left\{  \left(  \rho_{0}-\rho_{1}\right)  \geq0\right\}  $. That is,
supposing that $\Pi_{+}\equiv\left\{  \left(  \rho_{0}-\rho_{1}\right)
\geq0\right\}  $ and $\Pi_{-}\equiv\left\{  \left(  \rho_{0}-\rho_{1}\right)
<0\right\}  $, one can write the two-bit test as the product of two controlled
gates%
\begin{align}
U_{1} &  \equiv I_{B_{1}}\otimes\left(  \Pi_{+}\right)  _{B_{2}}\otimes
I_{A}+I_{B_{1}}\otimes\left(  \Pi_{-}\right)  _{B_{2}}\otimes\left(
\sigma_{X}\right)  _{A},\\
U_{2} &  \equiv\left(  \Pi_{+}\right)  _{B_{1}}\otimes I_{B_{2}}\otimes
I_{A}+\left(  \Pi_{-}\right)  _{B_{1}}\otimes I_{B_{2}}\otimes\left(
\sigma_{X}\right)  _{A},
\end{align}
where $B_{1}$ is the first channel output, $B_{2}$ is the second channel
output, and $A$ is an ancillary system initialized to the state $\left\vert
0\right\rangle $. The product of these two unitary gates is equal to%
\begin{multline}
U_{1}U_{2}=\left(  \left(  \Pi_{+}\right)  _{B_{1}}\otimes\left(  \Pi
_{+}\right)  _{B_{2}}+\left(  \Pi_{-}\right)  _{B_{1}}\otimes\left(  \Pi
_{-}\right)  _{B_{2}}\right)  \otimes I_{A}+\\
\left(  \left(  \Pi_{-}\right)  _{B_{1}}\otimes\left(  \Pi_{+}\right)
_{B_{2}}+\left(  \Pi_{+}\right)  _{B_{1}}\otimes\left(  \Pi_{-}\right)
_{B_{2}}\right)  \otimes\left(  \sigma_{X}\right)  _{A}.
\end{multline}
The receiver would then measure the ancillary system $A$ in order to make a
decision about~$u_{1}$.

Next, we determine the decoding of $u_{2}$, given that $u_{1}$ has
already been decoded. By the definition of the polar encoder transformation in
(\ref{eq:first-polar-encoder}), the goal is to distinguish between the
following two states:%
\begin{equation}
\rho_{u_{1}}\otimes\rho_{0},\ \ \ \ \ \ \ \rho_{u_{1}+1}\otimes\rho_{1}.
\end{equation}
The optimal quantum hypothesis test is given by the following projector:%
\begin{equation}
\left\{  \rho_{u_{1}}\otimes\rho_{0}-\rho_{u_{1}+1}\otimes\rho_{1}%
\geq0\right\}  .\label{eq:two-qubit-two-parity}%
\end{equation}
This optimal quantum hypothesis test is not factorizable into smaller tests,
and indeed, it is necessary to perform a collective measurement in order to
implement it. Nonetheless, Lemma~\ref{lem:FC-error-N-copies}\ provides a simple
implementation of the Fuchs-Caves measurement for distinguishing these two states.
% Note that in the case $u_1=0$, the protocol of  \cite{BCZ} can also be used. 

\subsection{Four-Bit Polar Decoder}

We now consider the form of Helstrom measurements for a four-bit polar code.
Recall that the input transformation for the four-bit polar code is as follows:%
\begin{equation}
\left(  u_{1},u_{2},u_{3},u_{4}\right)  \rightarrow\left(  u_{1}+u_{2}%
+u_{3}+u_{4},u_{3}+u_{4},u_{2}+u_{4},u_{4}\right)  .
\end{equation}
It is straightforward to find the form of the four different tests for
decoding $u_{1}$ through $u_{4}$. (See the appendix for derivations.) The test
for decoding $u_{1}$ is again a parity test:%
\begin{equation}
\left\{  \left(  \rho_{0}-\rho_{1}\right)  ^{\otimes4}\geq0\right\}  .
\end{equation}
The test for decoding $u_{2}$ given $u_{1}$ is%
\begin{equation}
\left\{  \left(  \sum_{u_{3}^{\prime}}\rho_{u_{1}+u_{3}^{\prime}}\otimes
\rho_{u_{3}^{\prime}}\right)  \otimes\left(  \sum_{u_{4}}\rho_{u_{4}}%
\otimes\rho_{u_{4}}\right)  -\left(  \sum_{u_{3}^{\prime}}\rho_{u_{1}%
+1+u_{3}^{\prime}}\otimes\rho_{u_{3}^{\prime}}\right)  \otimes\left(
\sum_{u_{4}}\rho_{1+u_{4}}\otimes\rho_{u_{4}}\right)  \geq0\right\}  .
\end{equation}
It remains unclear to us if there is a simple way to decompose the above test
any further into non-collective actions (or even approximately using, e.g., the
Fuchs-Caves measurement).
The test for decoding $u_{3}$ given $u_{2}$ and $u_{1}$ is%
\begin{equation}
\left\{  \left(  \rho_{u_{1}+u_{2}}\otimes\rho_{0}-\rho_{u_{1}+u_{2}+1}%
\otimes\rho_{1}\right)  \otimes\left(  \rho_{u_{2}}\otimes\rho_{0}-\rho
_{u_{2}+1}\otimes\rho_{1}\right)  \geq0\right\}  .
\end{equation}
One could actually approximate this test \textquotedblleft
efficiently\textquotedblright\ by performing a product Fuchs-Caves measurement
of the first two systems, a product Fuchs-Caves measurement of the last two,
and then take the parity of the results of these two tests (of course
implementing these tests coherently). The final Helstrom test for decoding
$u_{4}$ given $u_{3}$, $u_{2}$, and $u_{1}$ is%
\begin{equation}
\left\{  \rho_{u_{1}+u_{2}+u_{3}}\otimes\rho_{u_{3}}\otimes\rho_{u_{2}}%
\otimes\rho_{0}-\rho_{u_{1}+u_{2}+u_{3}+1}\otimes\rho_{u_{3}+1}\otimes
\rho_{u_{2}+1}\otimes\rho_{1}\geq0\right\}  .
\end{equation}
Clearly, it would be better to perform this last test 
by processing the likelihood ratios
resulting from individual
Fuchs-Caves measurements, rather than performing the optimal collective
Helstrom measurement.

\subsection{Polar Decoder for Larger Blocklengths}

One can continue in the above fashion to determine the form of
a quantum successive cancellation decoder that recovers each bit of an
eight-bit polar code.  We again try to simplify each Helstrom measurement and 
provide an expression for each one in Appendix \ref{8bit measurements}.
A few tests simplify, in particular those used to
recover the first bit $u_1$ (Eq.~\eqref{8bit_u1}), the fifth bit $u_5$
(Eq.~\eqref{8bit_u5}, the seventh bit $u_7$ (Eq.~\eqref{8bit_u7}),
and the last bit $u_8$ (Eq.~\eqref{8bit_u8}). However,
for the other tests, it is unclear if they can be 
approximated by some combination of Helstrom and
Fuchs-Caves measurements, followed by coherent post-processing.

From considering the eight-bit polar decoder, we can make several
observations. For any blocklength, it is always possible to recover the first
bit efficiently by calculating the parity of individual Helstrom measurements
(though, this bit is always the \textquotedblleft worst\textquotedblright%
\ bit, so the receiver would never actually be decoding it in practice). The receiver can always
 recover the last bit by performing a Fuchs-Caves measurement (this
is always the \textquotedblleft best\textquotedblright\ bit, so this should
already be evident from the main observation in this paper). Furthermore,
there are many bits that can be recovered by first performing Fuchs-Caves
measurements, followed by the parity of these tests. Unfortunately, the fraction of
these tests tends to zero in the limit of large blocklength. Thus, there still
remains much to understand regarding the structure of a polar decoder.

\section{Conclusion}

The main result of this paper is an advance over previous schemes for decoding
classical information transmitted over channels with classical inputs and
quantum outputs. In particular, we have shown that $N\cdot I\left(
W_{\text{acc}}\right)  $ of the information bits can be decoded reliably and
efficiently on a quantum computer by a \textquotedblleft
non-collective\textquotedblright\ coherent decoding strategy, while closing
the gap to the Holevo information rate (decoding the other $N \left(  I\left(
W\right)  -I\left(  W_{\text{acc}}\right)  \right)  $ bits) should require a
collective strategy. For the pure-loss bosonic channel, this implies that the
majority of the bits transmitted can be decoded by a product strategy whenever
the mean photon number is larger than one, while the fraction of collective
measurements required increases sharply as the mean photon number decreases
below one, marking the beginning of the quantum regime. Remarkably, even at
mean photon numbers as low as $10^{-8}$, roughly 10\% of the bits do not require
collective decoding, however. As another contribution, we have
shown that a receiver can also employ collective Fuchs-Caves measurements when
decoding a classical-quantum polar code. Finally, we gave the explicit form of
the Helstrom measurements of a quantum successive cancellation decoder for
two-, four-, and eight-bit polar codes. This should be helpful in determining the explicit
form of tests for larger blocklength polar codes.

The main open question is still to determine whether all of the information
bits can be efficiently decoded on a quantum computer. To answer this
question, one might consider employing the Schur transform \cite{BCZ, H05,C06} and
exploiting the structure inherent in polar codes. Unfortunately, it is not clear to
us that this approach will lead to a quantum successive cancellation decoder
with time complexity $O\left(  N\log N\right)  $ because the complexity of the
Schur transform is higher than this.

We acknowledge helpful discussions with Fr\'{e}d\'{e}ric Dupuis, Saikat Guha,
Hari Krovi, David Poulin, and Joseph Renes. MMW\ acknowledges support
from Montreal's Centre de Recherches Math\'{e}matiques. OLC aknowledges support
from NSERC through a Vanier scholarship. PH\ acknowledges support from the
Canada Research Chairs program, the Perimeter Institute, CIFAR, FQRNT's
INTRIQ, NSERC, and ONR through grant N000140811249.

\bibliography{Ref}

\appendix

\section{Derivations for the Four-Bit Polar Decoder Measurements}

The four-bit polar encoder amounts to the following transformation:%
\begin{equation}
\left(  u_{1},u_{2},u_{3},u_{4}\right)  \rightarrow\left(  u_{1}+u_{2}%
+u_{3}+u_{4},u_{3}+u_{4},u_{2}+u_{4},u_{4}\right)  .
\end{equation}

\subsection{Recovering $u_{1}$}

Let us first determine how the quantum successive cancellation decoder (QSCD)
recovers the bit $u_{1}$, assuming that $u_{2}$, $u_{3}$, and $u_{4}$ are
chosen uniformly at random. The test aims to distinguish between the following
two states:%
\begin{align}
&  \frac{1}{2^{3}}\sum_{u_{2},u_{3},u_{4}}\rho_{u_{2}+u_{3}+u_{4}}\otimes
\rho_{u_{3}+u_{4}}\otimes\rho_{u_{2}+u_{4}}\otimes\rho_{u_{4}},\\
&  \frac{1}{2^{3}}\sum_{u_{2},u_{3},u_{4}}\rho_{u_{2}+u_{3}+u_{4}+1}%
\otimes\rho_{u_{3}+u_{4}}\otimes\rho_{u_{2}+u_{4}}\otimes\rho_{u_{4}},
\end{align}
and it performs the following projection:%
\begin{align}
&  \left\{  \sum_{u_{2},u_{3},u_{4}}\left(  \rho_{u_{2}+u_{3}+u_{4}}%
-\rho_{u_{2}+u_{3}+u_{4}+1}\right)  \otimes\rho_{u_{3}+u_{4}}\otimes
\rho_{u_{2}+u_{4}}\otimes\rho_{u_{4}}\geq0\right\}  \nonumber\\
&  =\left\{  \sum_{u_{2},u_{3},u_{4}}\left(  -1\right)  ^{u_{2}+u_{3}+u_{4}%
}\left(  \rho_{0}-\rho_{1}\right)  \otimes\rho_{u_{3}+u_{4}}\otimes\rho
_{u_{2}+u_{4}}\otimes\rho_{u_{4}}\geq0\right\}  \\
&  =\left\{  \left(  \rho_{0}-\rho_{1}\right)  \otimes\sum_{u_{2},u_{3},u_{4}%
}\left(  -1\right)  ^{u_{2}+u_{3}+u_{4}}\rho_{u_{3}+u_{4}}\otimes\rho
_{u_{2}+u_{4}}\otimes\rho_{u_{4}}\geq0\right\}  \\
&  =\left\{  \left(  \rho_{0}-\rho_{1}\right)  \otimes\sum_{u_{2},u_{3},u_{4}%
}\left(  -1\right)  ^{u_{3}+u_{4}}\rho_{u_{3}+u_{4}}\otimes\left(  -1\right)
^{u_{2}+u_{4}}\rho_{u_{2}+u_{4}}\otimes\left(  -1\right)  ^{u_{4}}\rho_{u_{4}%
}\geq0\right\}
\end{align}%
\begin{align}
&  =\left\{  \left(  \rho_{0}-\rho_{1}\right)  \otimes\sum_{u_{2}^{\prime
},u_{3}^{\prime},u_{4}^{\prime}}\left(  -1\right)  ^{u_{2}^{\prime}}%
\rho_{u_{2}^{\prime}}\otimes\left(  -1\right)  ^{u_{3}^{\prime}}\rho
_{u_{3}^{\prime}}\otimes\left(  -1\right)  ^{u_{4}^{\prime}}\rho
_{u_{4}^{\prime}}\geq0\right\}  \\
&  =\left\{  \left(  \rho_{0}-\rho_{1}\right)  \otimes\sum_{u_{2}^{\prime}%
}\left(  -1\right)  ^{u_{2}^{\prime}}\rho_{u_{2}^{\prime}}\otimes\sum
_{u_{3}^{\prime}}\left(  -1\right)  ^{u_{3}^{\prime}}\rho_{u_{3}^{\prime}%
}\otimes\sum_{u_{4}^{\prime}}\left(  -1\right)  ^{u_{4}^{\prime}}\rho
_{u_{4}^{\prime}}\geq0\right\}  \\
&  =\left\{  \left(  \rho_{0}-\rho_{1}\right)  \otimes\left(  \rho_{0}%
-\rho_{1}\right)  \otimes\left(  \rho_{0}-\rho_{1}\right)  \otimes\left(
\rho_{0}-\rho_{1}\right)  \geq0\right\}  .
\end{align}
Thus, this first test nicely factors as the parity of the four individual
tests $\left\{  \left(  \rho_{0}-\rho_{1}\right)  \geq0\right\}  $.

\subsection{Recovering $u_{2}$ given $u_{1}$}

We now determine how the quantum successive cancellation decoder recovers
$u_{2}$ given $u_{1}$, while randomizing over $u_{3}$ and $u_{4}$. The aim is
to distinguish between the following two states:%
\begin{align}
&  \frac{1}{2^{2}}\sum_{u_{3},u_{4}}\rho_{u_{1}+u_{3}+u_{4}}\otimes\rho
_{u_{3}+u_{4}}\otimes\rho_{u_{4}}\otimes\rho_{u_{4}},\\
&  \frac{1}{2^{2}}\sum_{u_{3},u_{4}}\rho_{u_{1}+1+u_{3}+u_{4}}\otimes
\rho_{u_{3}+u_{4}}\otimes\rho_{1+u_{4}}\otimes\rho_{u_{4}},
\end{align}
which translates to a projection of the following form:%
\begin{equation}
\left\{  \sum_{u_{3},u_{4}}\rho_{u_{1}+u_{3}+u_{4}}\otimes\rho_{u_{3}+u_{4}%
}\otimes\rho_{u_{4}}\otimes\rho_{u_{4}}-\rho_{u_{1}+1+u_{3}+u_{4}}\otimes
\rho_{u_{3}+u_{4}}\otimes\rho_{1+u_{4}}\otimes\rho_{u_{4}}\geq0\right\}  .
\end{equation}
Define $u_{3}^{\prime}=u_{3}+u_{4}$ and the above becomes%
\begin{align}
&  \left\{  \sum_{u_{3}^{\prime},u_{4}}\rho_{u_{1}+u_{3}^{\prime}}\otimes
\rho_{u_{3}^{\prime}}\otimes\rho_{u_{4}}\otimes\rho_{u_{4}}-\rho
_{u_{1}+1+u_{3}^{\prime}}\otimes\rho_{u_{3}^{\prime}}\otimes\rho_{1+u_{4}%
}\otimes\rho_{u_{4}}\geq0\right\}  \nonumber\\
&  =\left\{  \left(  \sum_{u_{3}^{\prime}}\rho_{u_{1}+u_{3}^{\prime}}%
\otimes\rho_{u_{3}^{\prime}}\right)  \otimes\left(  \sum_{u_{4}}\rho_{u_{4}%
}\otimes\rho_{u_{4}}\right)  -\left(  \sum_{u_{3}^{\prime}}\rho_{u_{1}%
+1+u_{3}^{\prime}}\otimes\rho_{u_{3}^{\prime}}\right)  \otimes\left(
\sum_{u_{4}}\rho_{1+u_{4}}\otimes\rho_{u_{4}}\right)  \geq0\right\}  .
\end{align}

\subsection{Recovering $u_{3}$ given $u_{2}$ and $u_{1}$}

Let us determine how the QSCD recovers $u_{3}$ given $u_{2}$ and $u_{1}$,
while randomizing over $u_{4}$. The test distinguishes between the following
two states:%
\begin{align}
&  \frac{1}{2}\sum_{u_{4}}\rho_{u_{1}+u_{2}+u_{4}}\otimes\rho_{u_{4}}%
\otimes\rho_{u_{2}+u_{4}}\otimes\rho_{u_{4}},\\
&  \frac{1}{2}\sum_{u_{4}}\rho_{u_{1}+u_{2}+1+u_{4}}\otimes\rho_{1+u_{4}%
}\otimes\rho_{u_{2}+u_{4}}\otimes\rho_{u_{4}},
\end{align}
and amounts to a projector of the following form:%
\begin{align}
&  \left\{  \sum_{u_{4}}\rho_{u_{1}+u_{2}+u_{4}}\otimes\rho_{u_{4}}\otimes
\rho_{u_{2}+u_{4}}\otimes\rho_{u_{4}}-\sum_{u_{4}}\rho_{u_{1}+u_{2}+1+u_{4}%
}\otimes\rho_{1+u_{4}}\otimes\rho_{u_{2}+u_{4}}\otimes\rho_{u_{4}}%
\geq0\right\}  \nonumber\\
&  =\left\{  \sum_{u_{4}}\left(  \rho_{u_{1}+u_{2}+u_{4}}\otimes\rho_{u_{4}%
}-\rho_{u_{1}+u_{2}+1+u_{4}}\otimes\rho_{1+u_{4}}\right)  \otimes\rho
_{u_{2}+u_{4}}\otimes\rho_{u_{4}}\geq0\right\}  \\
&  =\left\{  \sum_{u_{4}}\left(  -1\right)  ^{u_{4}}\left(  \rho_{u_{1}+u_{2}%
}\otimes\rho_{0}-\rho_{u_{1}+u_{2}+1}\otimes\rho_{1}\right)  \otimes
\rho_{u_{2}+u_{4}}\otimes\rho_{u_{4}}\geq0\right\}  \\
&  =\left\{  \left(  \rho_{u_{1}+u_{2}}\otimes\rho_{0}-\rho_{u_{1}+u_{2}%
+1}\otimes\rho_{1}\right)  \otimes\sum_{u_{4}}\left(  -1\right)  ^{u_{4}}%
\rho_{u_{2}+u_{4}}\otimes\rho_{u_{4}}\geq0\right\}  \\
&  =\left\{  \left(  \rho_{u_{1}+u_{2}}\otimes\rho_{0}-\rho_{u_{1}+u_{2}%
+1}\otimes\rho_{1}\right)  \otimes\left(  \rho_{u_{2}}\otimes\rho_{0}%
-\rho_{u_{2}+1}\otimes\rho_{1}\right)  \geq0\right\}  .
\end{align}
Thus, this test nicely factorizes as the parity of two tests $\left\{  \left(
\rho_{u_{1}+u_{2}}\otimes\rho_{0}-\rho_{u_{1}+u_{2}+1}\otimes\rho_{1}\right)
\geq0\right\}  $ and $\left\{  \left(  \rho_{u_{2}}\otimes\rho_{0}-\rho
_{u_{2}+1}\otimes\rho_{1}\right)  \geq0\right\}  $.

\subsection{Recovering $u_{4}$ given $u_{3}$, $u_{2}$, and $u_{1}$}

Finally, we determine how the QSCD recovers $u_{4}$ given all of the previous
bits. The test in this case just aims to distinguish the following states:%
\begin{align}
&  \rho_{u_{1}+u_{2}+u_{3}}\otimes\rho_{u_{3}}\otimes\rho_{u_{2}}\otimes
\rho_{0},\\
&  \rho_{u_{1}+u_{2}+u_{3}+1}\otimes\rho_{u_{3}+1}\otimes\rho_{u_{2}+1}%
\otimes\rho_{1},
\end{align}
and amounts to the following projection:%
\begin{equation}
\left\{  \rho_{u_{1}+u_{2}+u_{3}}\otimes\rho_{u_{3}}\otimes\rho_{u_{2}}%
\otimes\rho_{0}-\rho_{u_{1}+u_{2}+u_{3}+1}\otimes\rho_{u_{3}+1}\otimes
\rho_{u_{2}+1}\otimes\rho_{1}\geq0\right\}
.\label{eq:last-four-qubit-decision}%
\end{equation}

\section{Measurements for the Eight-Bit Polar Decoder \label{8bit measurements}}

Here, we provide the form of
a quantum successive cancellation decoder that recovers each bit of an
eight-bit polar code. Full derivations of the results in this section
are available from the authors upon request.

\subsection{Recovering $u_{1}$}

The test to recover
the first bit $u_{1}$ is simply the parity of eight individual Helstrom
measurements:%
\begin{equation}
\left\{  \left(  \rho_{0}-\rho_{1}\right)  ^{\otimes8}\geq0\right\}  .
\label{8bit_u1}
\end{equation}

\subsection{Recovering $u_{2}$ given $u_{1}$}

The test to recover bit $u_{2}$ given $u_{1}$ projects onto the
positive eigenspace of the difference of%
\begin{multline}
\left(  \sum_{u_{3}^{\prime},u_{4}^{\prime},u_{5}^{\prime}}%
\rho_{u_{1}+u_{3}^{\prime}+u_{4}^{\prime}+u_{5}^{\prime}}\otimes\rho
_{u_{3}^{\prime}}\otimes\rho_{u_{4}^{\prime}}\otimes\rho_{u_{5}^{\prime}%
}\right)  
\otimes\left(  \sum_{u_{6}^{\prime},u_{7}^{\prime}%
,u_{8}^{\prime}}\rho_{u_{6}^{\prime}+u_{7}^{\prime}+u_{8}^{\prime}}\otimes
\rho_{u_{6}^{\prime}}\otimes\rho_{u_{7}^{\prime}}\otimes\rho_{u_{8}^{\prime}%
}\right)
\end{multline}
and%
\begin{multline}
\left(  \sum_{u_{3}^{\prime},u_{4}^{\prime},u_{5}^{\prime}}%
\rho_{u_{1}+u_{3}^{\prime}+u_{4}^{\prime}+u_{5}^{\prime}}\otimes\rho
_{u_{3}^{\prime}}\otimes\rho_{u_{4}^{\prime}}\otimes\rho_{u_{5}^{\prime}%
}\right)  
\otimes\left(  \sum_{u_{6}^{\prime},u_{7}^{\prime}%
,u_{8}^{\prime}}\rho_{u_{6}^{\prime}+u_{7}^{\prime}+u_{8}^{\prime}}\otimes
\rho_{u_{6}^{\prime}}\otimes\rho_{u_{7}^{\prime}}\otimes\rho_{u_{8}^{\prime}%
}\right)  .
\label{8bit_u2}
\end{multline}
As such, it is not clear to us how one could approximate this test as some
combination of Helstrom and Fuchs-Caves tests. 

\subsection{Recovering $u_{3}$ given $u_{2}$, and $u_{1}$}

The test to recover bit $u_{3}$
given $u_{1}$ and $u_{2}$ is equal to the parity of the following two tests:%
\begin{align}
& \left\{
\begin{array}
[c]{c}%
\left(  \sum_{u_{4}^{\prime}}\rho_{u_{1}+u_{2}+u_{4}^{\prime}}\otimes
\rho_{u_{4}^{\prime}}\right)  \otimes\left(  \sum_{u_{5}^{\prime}}\rho
_{u_{5}^{\prime}}\otimes\rho_{u_{5}^{\prime}}\right)  \\
-\left(  \sum_{u_{4}^{\prime}}\rho_{u_{1}+u_{2}+1+u_{4}^{\prime}}\otimes
\rho_{u_{4}^{\prime}}\right)  \otimes\left(  \sum_{u_{5}^{\prime}}%
\rho_{1+u_{5}^{\prime}}\otimes\rho_{u_{5}^{\prime}}\right)  \geq0
\end{array}
\right\}  ,\\
& \left\{
\begin{array}
[c]{c}%
\left(  \sum_{u_{6}^{\prime}}\rho_{u_{2}+u_{6}^{\prime}}\otimes\rho
_{u_{6}^{\prime}}\right)  \otimes\left(  \sum_{u_{8}^{\prime\prime}}%
\rho_{u_{8}^{\prime\prime}}\otimes\rho_{u_{8}^{\prime\prime}}\right)  \\
-\left(  \sum_{u_{6}^{\prime}}\rho_{u_{2}+u_{6}^{\prime}+1}\otimes\rho
_{u_{6}^{\prime}}\right)  \otimes\left(  \sum_{u_{8}^{\prime\prime}}%
\rho_{1+u_{8}^{\prime\prime}}\otimes\rho_{u_{8}^{\prime\prime}}\right)  \geq0
\end{array}
\right\}  .
\label{8bit_u3}
\end{align}
It is again unclear to us how to decompose this measurement further. 

\subsection{Recovering $u_{4}$ given $u_{3}$, $u_{2}$, and $u_{1}$}

The test
to recover bit $u_{4}$ given $u_{1}$, $u_{2}$, and $u_{3}$
projects onto the positive eigenspace of the difference of%
\begin{multline}
\left(  \sum_{u_{5}^{\prime}}\rho_{u_{1}+u_{2}+u_{3}+u_{5}^{\prime}}%
\otimes\rho_{u_{5}^{\prime}}\right)  \otimes\left(  \sum_{u_{6}^{\prime}}%
\rho_{u_{3}+u_{6}^{\prime}}\otimes\rho_{u_{6}^{\prime}}\right)  
\otimes\left(  \sum_{u_{7}^{\prime}}\rho_{u_{2}+u_{7}^{\prime}}\otimes
\rho_{u_{7}^{\prime}}\right)  \otimes\left(  \sum_{u_{8}^{\prime}}\rho
_{u_{8}^{\prime}}\otimes\rho_{u_{8}^{\prime}}\right)
\end{multline}
and%
\begin{multline}
\left(  \sum_{u_{5}^{\prime}}\rho_{u_{1}+u_{2}+u_{3}+1+u_{5}^{\prime}}%
\otimes\rho_{u_{5}^{\prime}}\right)  \otimes\left(  \sum_{u_{6}^{\prime}}%
\rho_{u_{3}+1+u_{6}^{\prime}}\otimes\rho_{u_{6}^{\prime}}\right)  
\otimes\left(  \sum_{u_{7}^{\prime}}\rho_{u_{2}+1+u_{7}^{\prime}}\otimes
\rho_{u_{7}^{\prime}}\right)  \otimes\left(  \sum_{u_{8}^{\prime}}%
\rho_{1+u_{8}^{\prime}}\otimes\rho_{u_{8}^{\prime}}\right)
\label{8bit_u4}
\end{multline}
Again, this one remains unclear how to decompose further. 

\subsection{Recovering $u_{5}$ given $u_{4}$, \ldots, $u_{1}$}

The test to recover
bit $u_{5}$ given $u_{1}$ through $u_{4}$ is equal to%
\begin{equation} 
\left\{
\begin{array}
[c]{c}%
\left(  \rho_{u_{1}+u_{2}+u_{3}+u_{4}}\otimes\rho_{0}-\rho_{u_{1}+u_{2}%
+u_{3}+u_{4}+1}\otimes\rho_{1}\right)  \otimes\left(  \rho_{u_{3}+u_{4}%
}\otimes\rho_{0}-\rho_{u_{3}+u_{4}+1}\otimes\rho_{1}\right)  \\
\otimes\left(  \rho_{u_{2}+u_{4}}\otimes\rho_{0}-\rho_{u_{2}+u_{4}+1}%
\otimes\rho_{1}\right)  \otimes\left(  \rho_{u_{4}}\otimes\rho_{0}-\rho
_{u_{4}+1}\otimes\rho_{1}\right)  \geq0
\end{array}
\right\}  .
\label{8bit_u5}
\end{equation}
It is easy to see that one could approximate this test by
first performing four Fuchs-Caves measurements
on adjacent pairs of channel outputs and taking the parity of these
tests. 

\subsection{Recovering $u_{6}$ given $u_{5}$, \ldots, $u_{1}$}

The test to recover bit $u_{6}$ given $u_{1}$ through $u_{5}$ is a
projection onto the positive eigenspace of the difference of%
\begin{multline}
\left(  \sum_{u_{7}^{\prime}}\rho_{u_{1}+\cdots+u_{5}+u_{7}%
^{\prime}}\otimes\rho_{u_{5}+u_{7}^{\prime}}\otimes\rho_{u_{3}+u_{4}%
+u_{7}^{\prime}}\otimes\rho_{u_{7}^{\prime}}\right)  
\otimes\left(  \sum_{u_{8}^{\prime}}\rho_{u_{2}+u_{4}+u_{8}^{\prime}}%
\otimes\rho_{u_{8}^{\prime}}\otimes\rho_{u_{4}+u_{8}^{\prime}}\otimes
\rho_{u_{8}^{\prime}}\right)
\end{multline}
and%
\begin{equation}
\left(  \sum_{u_{7}^{\prime}}\rho_{u_{1}+\cdots+u_{5}%
+1+u_{7}^{\prime}}\otimes\rho_{u_{5}+1+u_{7}^{\prime}}\otimes\rho_{u_{3}%
+u_{4}+u_{7}^{\prime}}\otimes\rho_{u_{7}^{\prime}}\right)  
\otimes\left(  \sum_{u_{8}^{\prime}}\rho_{u_{2}+u_{4}+1+u_{8}^{\prime}}%
\otimes\rho_{1+u_{8}^{\prime}}\otimes\rho_{u_{4}+u_{8}^{\prime}}\otimes
\rho_{u_{8}^{\prime}}\right)  .
\label{8bit_u6}
\end{equation}
A simple decomposition of this test remains unclear. 

\subsection{Recovering $u_{7}$ given $u_{6}$, \ldots, $u_{1}$}

The test for recovering
bit $u_{7}$ given the previous ones is%
\begin{equation}
\left\{
\begin{array}
[c]{c}%
\left(  \rho_{u_{1}+\cdots +u_{6}}\otimes\rho_{u_{5}+u_{6}}\otimes\rho
_{u_{3}+u_{4}}\otimes\rho_{0}-\rho_{u_{1}+\cdots+u_{6}+1}\otimes\rho
_{u_{5}+u_{6}+1}\otimes\rho_{u_{3}+u_{4}+1}\otimes\rho_{1}\right)  \otimes\\
\left(  \rho_{u_{2}+u_{4}+u_{6}}\otimes\rho_{u_{6}}\otimes\rho_{u_{4}}%
\otimes\rho_{0}-\rho_{u_{2}+u_{4}+u_{6}+1}\otimes\rho_{u_{6}+1}\otimes
\rho_{u_{4}+1}\otimes\rho_{1}\right)  \geq0
\end{array}
\right\}  ,
\label{8bit_u7}
\end{equation}
which is clearly implementable by performing a Fuchs-Caves measurement on the
first four qubits and the last four, and than taking the parity of these two
tests. 

\subsection{Recovering $u_{8}$ given $u_{7}$, \ldots, $u_{1}$}

The final test for recovering the last bit $u_{8}$ given all others is
a projection onto the positive eigenspace of the difference of%
\begin{equation}
\rho_{u_{1}+\cdots+u_{7}}\otimes\rho_{u_{5}+u_{6}%
+u_{7}}\otimes\rho_{u_{3}+u_{4}+u_{7}}\otimes\rho_{u_{7}}
\otimes\rho_{u_{2}+u_{4}+u_{6}}\otimes\rho_{u_{6}}\otimes\rho_{u_{4}}%
\otimes\rho_{0},
\end{equation}
and%
\begin{equation}
\rho_{u_{1}+\cdots +u_{7}+1}\otimes\rho_{u_{5}%
+u_{6}+u_{7}+1}\otimes\rho_{u_{3}+u_{4}+u_{7}+1}\otimes\rho_{u_{7}+1}
\otimes\rho_{u_{2}+u_{4}+u_{6}+1}\otimes\rho_{u_{6}+1}\otimes\rho_{u_{4}%
+1}\otimes\rho_{1}.
\label{8bit_u8}
\end{equation}
It is clear that we can approximate this test with a Fuchs-Caves measurement.

\end{document}